\documentclass[12pt,reqno]{amsart}

\usepackage{amssymb,units,url}
\usepackage[bookmarks=false,unicode,colorlinks,urlcolor=red,citecolor=red,linkcolor=blue]{hyperref}
\usepackage{doi}
\usepackage{aliascnt}

\usepackage{pgfplots} 

\usepackage[margin=1.25in]{geometry}

\newtheorem*{theorem*}{Theorem}

\newtheorem{theorem}{Theorem}[section]

\newaliascnt{lemma}{theorem}
\newtheorem{lemma}[lemma]{Lemma}
\aliascntresetthe{lemma}

\newaliascnt{proposition}{theorem}
\newtheorem{proposition}[proposition]{Proposition}
\aliascntresetthe{proposition}

\newaliascnt{corollary}{theorem}
\newtheorem{corollary}[corollary]{Corollary}
\aliascntresetthe{corollary}

\newaliascnt{conjecture}{theorem}
\newtheorem{conjecture}[conjecture]{Conjecture}
\aliascntresetthe{conjecture}

\newaliascnt{problem}{theorem}
\newtheorem{problem}[problem]{Problem}
\aliascntresetthe{problem}

\newaliascnt{example}{theorem}

\aliascntresetthe{example}

\makeatletter
\def\tagform@#1{\maketag@@@{\ignorespaces#1\unskip\@@italiccorr}}
\let\orgtheequation\theequation
\def\theequation{(\orgtheequation)}
\makeatother
\def\equationautorefname~{}

\newcommand{\R}{{\mathbb R}}
\newcommand{\Rn}{{\mathbb R}^n}
\newcommand{\tr}{\operatorname{tr}}
\newcommand{\logV}{\operatorname{V_{\log}}}
\newcommand{\logcap}{\operatorname{C_{\log}}}

\begin{document}

\title[Minimizing capacity among linear images]{Minimizing capacity among linear images of rotationally invariant conductors}
\author[]{Richard S. Laugesen}
\address{Department of Mathematics, University of Illinois, Urbana,
IL 61801, U.S.A.}
\email{Laugesen\@@illinois.edu}
\date{\today. \quad ORCID: 0000-0003-1106-7203}

\subjclass[2010]{\text{Primary 31B15. Secondary 31A15, 35J05}}
\keywords{Isoperimetric, Riesz kernel, shape optimization.}

\begin{abstract}
Logarithmic capacity is shown to be minimal for a planar set having $N$-fold rotational symmetry ($N \geq 3$),  among all conductors obtained from the set by area-preserving linear transformations. Newtonian and Riesz capacities obey a similar property in all dimensions, when suitably normalized linear transformations are applied to a set having irreducible symmetry group. A corollary is P\'{o}lya and Schiffer's lower bound on capacity in terms of moment of inertia. 
\end{abstract}

\maketitle

\section{\bf Introduction and results}

Optimality and symmetry have been intertwined since ancient times, notably in the isoperimetric theorem, which implies that the region in the plane with given area and shortest perimeter is a disk. Lord Rayleigh introduced isoperimetry into mathematical physics with his conjecture that the membrane having lowest tone of vibration must be circular, a claim proved fifty years later by Faber and Krahn. 

Poincar\'{e} raised an analogous conjecture for the electrostatic capacity of a conductor. The resulting Poincar\'{e}--Carleman--Szeg\H{o} theorem \cite{PS51} asserts that the conductor of given volume that minimizes the Newtonian capacity is the set possessing the greatest possible rotational symmetry, namely the ball. 

One might wonder what happens in situations where only partial rotational symmetry is achievable, such as in the class of polygonal domains.   
\begin{theorem*}[Solynin and Zalgaller \cite{SZ04}]
Among all $N$-sided polygons with given area, the shape minimizing the logarithmic capacity is the regular $N$-gon. 
\end{theorem*}
Thus optimality again occurs for the most symmetrical shape within the class of competitors. Incidentally, the corresponding problem for the fundamental tone of the Laplacian remains open for $N \geq 5$, where it is called the polygonal Rayleigh--Faber--Krahn problem. Cheeger constants and variational energies have been considered on polygons too, in recent work by Bucur and Fragal\`{a} \cite{BF16,BF21}.

This paper establishes optimality results of a similar flavor for capacity in all dimensions, among the class of linear images of rotationally symmetric shapes. 

\subsection*{\bf Results for logarithmic capacity}

The logarithmic energy of a compact set $K \subset \Rn, n \geq 2$, is 
\[
\logV(K) = \min_\mu \int_K \int_K \log \frac{1}{|x-y|} \, d\mu(x) d\mu(y) 
\]
where the minimum is taken over all Borel probability measures (unit measures) on $K$. The minimum is attained for a measure $\mu$ called the logarithmic equilibrium measure of $K$. Boundedness of $K$ implies $-\infty < \logV(K) \leq +\infty$. If $\logV(K)<\infty$ then the equilibrium measure is unique. References for these foundational facts are provided in \autoref{sec-background}.

The logarithmic capacity is defined to be 
\[
\logcap(K) = \exp \left( -\logV(K) \right) ,
\]
so that $0 \leq \logcap(K) < \infty$. Note that if $K$ has positive area then it has finite energy (by choosing $\mu$ to be normalized area measure on $K$) and hence has positive capacity. 

A group $\mathcal{U}$ of orthogonal $n \times n$ matrices is called irreducible if the only subspaces of $\Rn$ that are invariant under the action of the group are $\Rn$ and the zero subspace. Equivalently, the group is irreducible if every nontrivial orbit spans $\Rn$, that is, if $\{ Ux : U \in \mathcal{U} \}$ spans $\Rn$ for every nonzero vector $x \in \Rn$. Examples of irreducible groups include the symmetry groups of regular $N$-gons in the plane, and of  platonic solids in $3$ dimensions. 

An isometry of a set $K \subset \Rn$ is an orthogonal $n \times n$ matrix $U$ such that $UK=K$. Call a square matrix $M$ volume-preserving if it has determinant $\pm 1$. 

Our first theorem minimizes logarithmic capacity among sets of the same volume.  
\begin{theorem}[Logarithmic capacity of linear images]  \label{capacitymin}
Let $K \subset \Rn, n \geq 2$, be a compact set with positive logarithmic capacity. If $M$ is a real, $n \times n$ volume-preserving matrix, and $K$ has an irreducible group of isometries, then 
\[
\logcap(MK) \geq \logcap(K)
\]
with equality if and only if $M$ is orthogonal. In particular, if $K$ has an irreducible isometry group and positive volume then it minimizes logarithmic capacity among all its linear images of the same volume. 
\end{theorem}
The proof is in \autoref{mainproof}. The idea in $2$ dimensions is to reduce to a diagonal matrix $M$ by the singular value decomposition, and then prove that the logarithmic energy of the set is strictly concave with respect to the $1$-parameter family $S_t = \left( \begin{smallmatrix} e^t & 0 \\ 0 & e^{-t} \end{smallmatrix} \right)$ of diagonal matrices applied to $K$ (\autoref{concavitypenergy}). The underlying point, when one strips away everything else in that proof, is convexity of the mapping
\[
t \mapsto \log |S_t x| = \frac{1}{2} \log (e^{2t}x_1^2+e^{-2t}x_2^2) , \qquad -\infty<t<\infty .
\]
Further, the logarithmic energy has vanishing first derivative at $t=0$ due to rotational symmetry of $K$ (\autoref{logcapfirstderivsymmetric}). Hence the maximal energy and thus minimal capacity are attained at $K$, as claimed in the theorem. 

These tools have some predecessors in the literature. Concavity,  superharmonicity and monotonicity results for potential theoretic energies under various deformations, transplantations and flows can be found for example in recent work by Betsakos et al.\ \cite{BKKP20} and Pouliasis \cite{P11,P21b} and in prior work by numerous authors \cite{B83,DK08,L93,PS53,S54,S92}, although none of that literature provides the tools needed in the current paper. 

In the planar case $n=2$, call $M$ area-preserving if its determinant is $\pm 1$. Since positive area implies positive capacity, \autoref{capacitymin} yields in that  case:
\begin{corollary}[Capacity of linear images of rotationally symmetric planar sets] \label{capacitymin2}
Let $K \subset \R^2$ be a compact planar set with positive logarithmic capacity. If $M$ is area-preserving and $K$ has rotational symmetry of order $N \geq 3$, then 
\[
\logcap (MK) \geq \logcap(K)
\]
with equality if and only if $M$ is orthogonal. In particular, if $K$ has $N$-fold rotational symmetry and positive area then it minimizes logarithmic capacity among all its linear images with the same area. 
\end{corollary}
An appealing consequence is that the equilateral triangle minimizes logarithmic capacity among all triangles of the same area. This special case goes back to P\'{o}lya and Szeg\H{o} \cite[p.{\,}158]{PS51}, whose proof is quite different, relying on Steiner symmetrization and the Dirichlet integral characterization of capacity rather than the logarithmic energy approach used in this paper. For other examples of rotationally symmetric sets to which the corollary applies, see \autoref{figsets}. 
\begin{figure}
\includegraphics[scale=0.3]{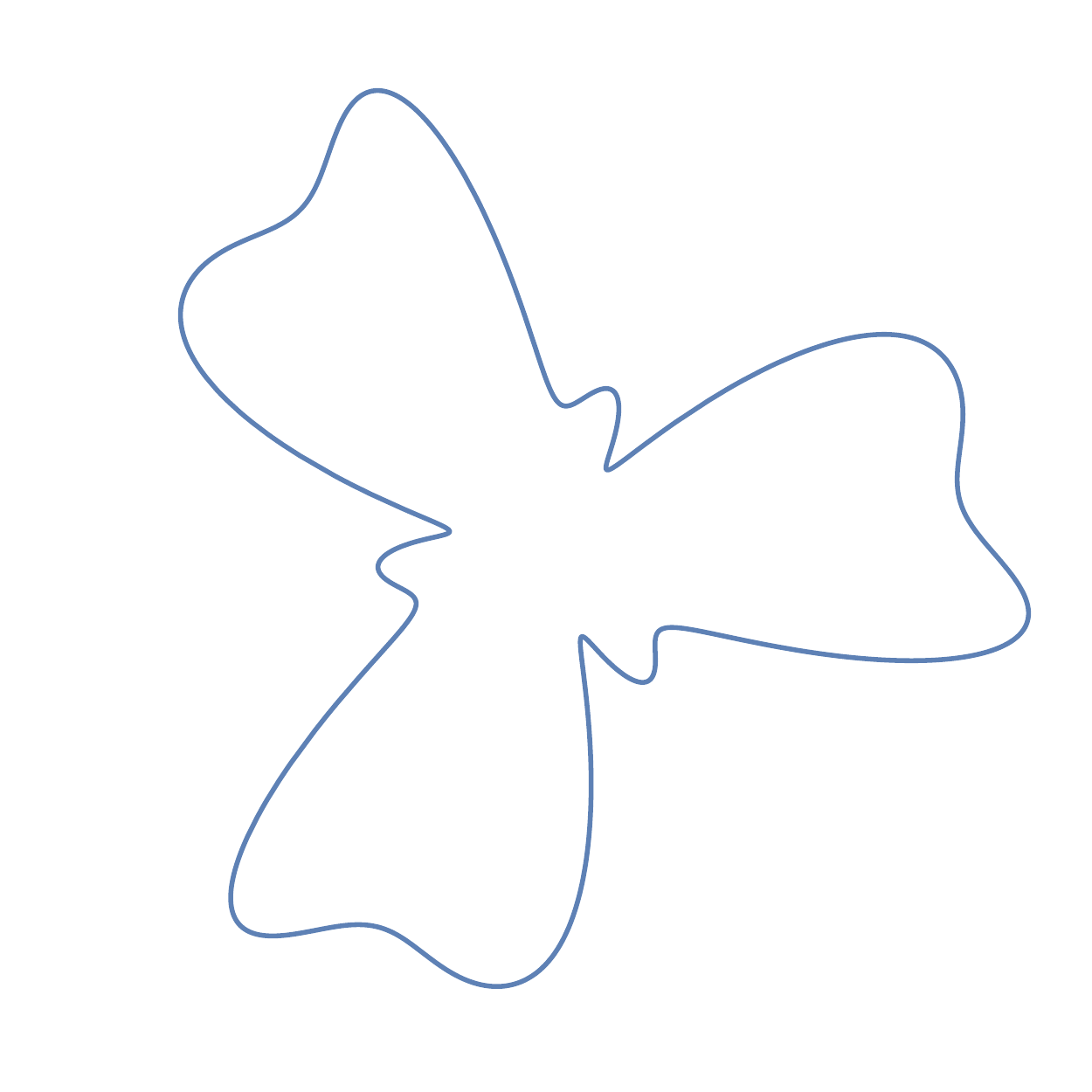} \qquad
\includegraphics[scale=0.3]{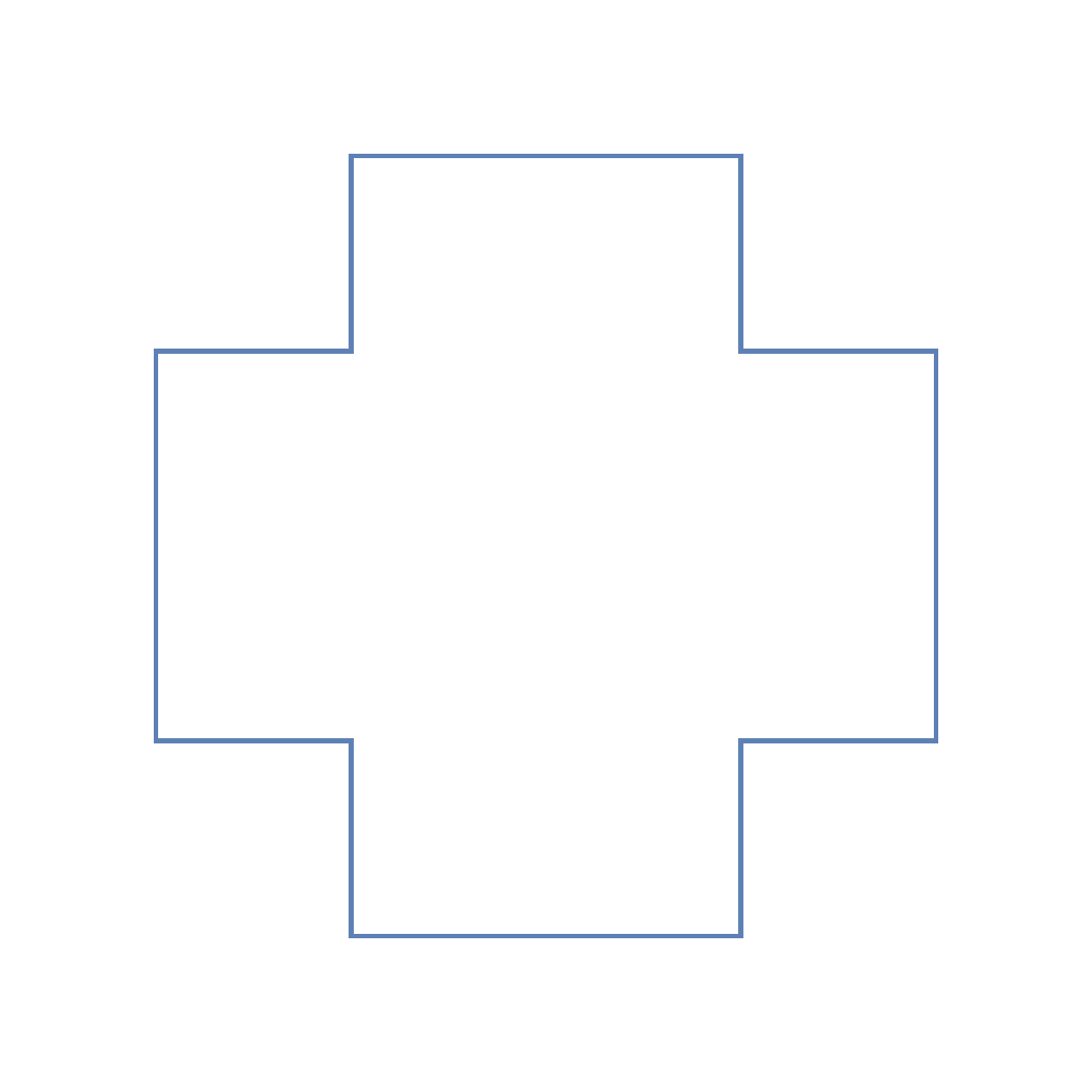} 
\caption{\label{figsets}Examples: compact sets with rotational symmetry of order $3,4$.}
\end{figure}

\autoref{capacitymin} and its corollary seem to be new, except that when $K$ is the regular $N$-gon, the corollary is a special case of Solynin and Zalgaller's theorem. 

Incidentally, a different type of minimization result for planar capacities of rotationally symmetric sets was proved by Baernstein and Solynin \cite[Sections 5,6]{BS13}. They take a circularly symmetric set contained in a sector of aperture $2\pi/N$ and compare capacities of certain unions of $\geq N$ rotations of this set with the union of rotations by $N$-th roots of unity. 

\subsection*{\bf Results for Riesz $p$-capacity}
The Riesz $p$-energy of a compact set $K \subset \Rn$ is 
\[
V_p(K) = \min_\mu \int_K \int_K \frac{1}{|x-y|^p} \, d\mu(x) d\mu(y) , \qquad 0<p<n ,
\]
where the minimum is taken over all Borel probability measures on $K$. In the literature, many authors call it the Riesz $\alpha$-energy where $\alpha=n-p$. 

The minimum is attained for a measure $\mu$ called the $p$-equilibrium measure of $K$. Boundedness of $K$ implies that the energy is positive: $0 < V_p(K) \leq +\infty$. If $V_p(K)<\infty$ then the equilibrium measure is unique. References for these facts are given in \autoref{sec-background}. 

The $p$-capacity is  
\[
C_p(K) = \frac{1}{V_p(K)^{1/p}} .
\]
Note that if the compact set $K$ has positive volume then it has finite energy (simply by choosing the measure $\mu$ in the infimum to be normalized Lebesgue measure on $K$, and using that $0<p<n$) and hence has positive capacity. 

The capacity is often defined to be $1/V_p(K)$, but we prefer the current definition because it ensures that  capacity scales linearly:  
\begin{equation} \label{eq:capacityscaling}
C_p(tK)=tC_p(K) , \qquad t > 0 .
\end{equation}
Some further intuition is gained by interpreting the logarithmic energy as the rate of change of $p$-energy in a formal asymptotic expansion at $p=0$:   
\[
V_p(K) = 1 + p \logV(K) + O(p^2) \quad \text{and hence} \quad 
C_p(K) = \logcap(K) (1+O(p)) 
\]
as $p \to 0$, by formally expanding the kernel $|x-y|^{-p}=\exp(p \log 1/|x-y|)$ in terms of the exponential series. Such asymptotic formulas are not needed for this paper. 

The special case $p=1$ is called the \textbf{Newtonian} situation in dimensions $n=2,3$, since the potential $1/|x-y|^p$ when $p=1$ takes the Newtonian form $1/r$.  

The minimization result that we state below for $p$-capacity normalizes not the volume of the linear image, but rather the $p$-norm of the singular values of the inverse matrix, as follows. Define the (Schatten) $p$-norm of a square $n \times n$ real matrix $A$ to be 
\[
\lVert A \rVert_{p,n} =  \left( \frac{1}{n} \sum_{k=1}^n \sigma_k(A)^p \right)^{\! \! 1/p} 
\]
where the $\sigma_k$ are the singular values of $A$. 
\begin{theorem}[$p$-capacity of linear images]  \label{capacityminp}
Let $K \subset \Rn, n \geq 2$, be a compact set with positive $p$-capacity, where $0<p<n$. If $M$ is a real, invertible $n \times n$ matrix with $\lVert M^{-1} \rVert_{p,n} = 1$, and $K$ has an irreducible group of isometries, then 
\[
C_p(MK) \geq C_p(K)
\]
with equality if and only if $M$ is orthogonal.   
\end{theorem}
See \autoref{mainproof} for the proof, which is analogous to the proof in the logarithmic case, but with new features that lead to the normalization in terms of the matrix $p$-norm. 

The theorem implies an old result of P\'{o}lya and Schiffer involving capacity and moment of inertia, as we proceed to explain. Suppose $K$ has positive volume $V(K)$ and let 
\[
  I(K) = \text{second moment of mass about the centroid} = \int_K |x-\overline{x}|^2 \, dx ,
\]
where the centroid is $\overline{x} = \int_K x \, dx/V(K)$. Recall that $I$ can be interpreted as an average moment of inertia with respect to a randomly chosen axis through the centroid (e.g.\ \cite[Proposition 7]{FLL07}). It can also be viewed as the moment of inertia of $K$ in $\R^{n+1}$ with respect to the axis perpendicular to $\Rn$ and passing through the centroid. 
\begin{corollary}[$p$-capacity normalized by moment of inertia and volume]\label{capmoment}
Let $K \subset \Rn, n \geq 2$, be a compact set with positive volume, and assume 
\[
\begin{cases}
0 < p < 2 & \text{if $n=2$,} \\
0 < p \leq \frac{2}{n-1} & \text{if $n \geq 3$.} 
\end{cases}
\]
If $K$ has an irreducible group of isometries then the scale invariant quantity 
\[
C_p \, \sqrt{\frac{I^{n-1}}{V^{n+1}}} 
\]
is minimal at $K$ among all nondegenerate linear images of that set, with equality if and only if the linear map is a euclidean similarity (a nonzero scalar multiple of an orthogonal transformation).

In particular, for the Newtonian ($p=1$) capacity of sets contained in $3$-dimensional space: \\
(i) if $K \subset \R^2$ has rotational symmetry of order $\geq 3$ and positive area then $C_1 \sqrt{I/A^3}$ is minimal at $K$ among all nondegenerate linear images; \\
(ii) if $K \subset \R^3$ has an irreducible group of isometries and positive volume then $C_1 I/V^2$ is minimal at $K$, among all nondegenerate linear images. 
\end{corollary}
The theorem is stronger than the corollary, in the sense that the derivation in \autoref{capmomentproof} invokes a H\"{o}lder inequality in order to estimate the $p$-norm of $M^{-1}$ in terms of the $2$-norm of $M$. 

P\'{o}lya and Schiffer \cite[(1.4) and {\S}4.4;  (1.11) and {\S}4.3]{PS53} proved the Newtonian special cases (i) and (ii) of \autoref{capmoment}. Their approach relies on the Thomson variational characterization of the Newtonian energy in terms of divergence-free fields. The stronger result \autoref{capacityminp} in this paper relies instead on the potential theoretic variational characterization in terms of kernels and measures.  

In the opposite direction, P\'{o}lya and Schiffer proved complementary upper bounds on capacity, for example that if $K \subset \R^2$ has rotational symmetry of order $\geq 3$ and positive area then $C_1 \sqrt{A/I}$ is maximal at $K$ among all nondegenerate linear images.

Next we examine the special case of $2$-capacity in \autoref{capacityminp}. The quantity
\begin{equation} \label{alphadef}
\alpha(K) = \sqrt{\frac{I(K)}{V(K)^{1+2/n}}}
\end{equation}
is scale invariant, and can be regarded as quantifying the asymmetry of $K$, since by an easy argument with mass transplantation, $\alpha(K)$ is minimal when $K$ is a ball. This asymmetry functional is different from Fraenkel's asymmetry, which appeared in the lower bounds on capacity by Hall, Hayman and Weitsman \cite{HHW91}. 

The next result restates \autoref{capacityminp} in  terms of capacity and asymmetry, for $p=2$. 
\begin{corollary}[$2$-capacity of linear images]  \label{capacityminn-2}
Let $K \subset \Rn, n \geq 3$, be a compact set with positive volume. If $K$ has an irreducible group of isometries, then 
\[
\alpha(M^{-1}K) \, \frac{C_2}{V^{1/n}} (MK) \geq \alpha(K) \, \frac{C_2}{V^{1/n}} (K) 
\]
whenever the matrix $M$ is invertible, with equality if and only if $M$ is a euclidean similarity. 
\end{corollary}
The corollary is proved in \autoref{corproofsec}. Note each of the four factors in the inequality is scale invariant with respect to replacing $K$ by $tK$, and with respect to replacing $M$ by $tM$, for $t>0$.  

An unusual feature of \autoref{capacityminn-2} is that the capacity is computed on one domain, $MK$, while the geometric asymmetry factor $\alpha(\cdot)$ is computed on an auxiliary domain, namely on $M^{-1}K$. The same phenomenon occurred some years ago in my joint work with Siudeja for Laplacian eigenvalue sums  on linear images of rotationally symmetric domains in higher dimensions \cite[Corollary 2]{LS11}. That one together with the earlier planar paper \cite{LS11a} helped inspire the present work on capacity. 

The corollary is discussed further in the next section. 

\section{\bf Open problems under volume constraint}

\subsection*{Volume normalization}
Do the $p$-capacity results in \autoref{capacityminp} and \autoref{capacityminn-2} hold under volume normalization, like the logarithmic capacity results in \autoref{capacitymin} and \autoref{capacitymin2}?
\begin{conjecture}[$p$-capacity of linear images under volume normalization] \label{convolume}
Let $n \geq 2$ and $0<p<n$. If $K$ is a compact set in $n$ dimensions that has an irreducible group of isometries then   
\[
C_p(MK) \geq C_p(K) 
\]
whenever $M$ is volume-preserving. In particular, if $K$ has positive volume then the scale invariant quantity $C_p/V^{1/n}$ is minimal at $K$ among all linear images:  
\[
\frac{C_p}{V^{1/n}} (MK) \geq \frac{C_p}{V^{1/n}}(K) .
\]
\end{conjecture}
The conjecture is stronger than \autoref{capacityminp}, as we now explain. The inequality $C_p(MK) \geq C_p(K)$ proved in \autoref{capacityminp} for an invertible matrix $M$ with $\lVert M^{-1} \rVert_{p,n} = 1$ is equivalent to having 
\[
C_p(MK) \lVert M^{-1} \rVert_{p,n} \geq C_p(K) 
\]
for all invertible $M$, by the scaling relation \eqref{eq:capacityscaling} for capacity. Meanwhile, \autoref{convolume} has scale invariant form
\[
C_p(MK) |\det M|^{-1/n} \geq C_p(K) .
\]
The connection is that Jensen's inequality applied to $s \mapsto \exp(ps)$ shows  
\begin{equation} \label{mnormdet}
\lVert M^{-1} \rVert_{p,n} \geq \left( \exp \Big( p \, \frac{1}{n} \sum_{k=1}^n \log \sigma_k(M^{-1}) \Big) \right)^{\! 1/p} = |\det M^{-1}|^{1/n} ,
\end{equation}
and so \autoref{convolume} implies \autoref{capacityminp}. 

Further, \autoref{convolume} would imply \autoref{capacityminn-2} directly, because linear maps increase the asymmetry of a symmetric set: 
\[
\frac{\alpha(M^{-1} K)}{\alpha(K) } = \frac{V(MK)^{1/n} }{V(K)^{1/n}} \, \lVert M^{-1} \rVert_{2,n} \geq 1 
\]
by using identity \eqref{2normidentity} and then inequality \eqref{mnormdet}. P\'{o}lya and Schiffer \cite[p.{\,}308]{PS53} already knew this fact about asymmetry, at least in dimension $2$ where they needed it. 

For $K$ a ball in dimension $n=3$, \autoref{convolume} can be verified numerically in the Newtonian case $p=1$ by means of the explicit formula for the capacity of an ellipsoid \cite[pp.{\,}429--431]{H89}, \cite[p.{\,}165]{L72}. Or, one could call on the more general Poincar\'{e}--Carleman--Szeg\H{o} theorem, which is discussed  below.

\subsection*{Simplices}
\autoref{convolume} would imply a striking new extremality property for simplices. In $3$ dimensions, for example, the $p$-capacity would be minimal for the regular tetrahedron among all tetrahedra having the same volume. In this statement, one uses that every tetrahedron is the image of the regular one under a linear transformation, up to translation. 

\subsection*{Does the ball minimize $p$-capacity?}
Turning our attention to arbitrary compact sets rather than linear images of a fixed set, we mention that minimality of $p$-capacity at the ball among compact sets of given volume is known by the Poincar\'{e}--Carleman--Szeg\H{o} theorem in the generalized Newtonian case ($p=n-2$). That theorem has been extended to the $\alpha$-stable process case ($n-2<p<n$) by Watanabe \cite[p.{\,}489]{W83}; see also Betsakos \cite{B04a,B04b} and M\'{e}ndez--Hern\'{a}ndez \cite{MH06}. Rearrangement and polarization methods underlie these results, along with probabilistic characterizations of the capacity. 

The case $0<p<n-2$ apparently remains open: 
\begin{problem}[Minimizing $p$-capacity of compact sets under volume normalization] \label{convolumegeneral}
Let $n \geq 3$ and $0<p<n-2$. If $K \subset \Rn$ is a compact set and $B$ is a closed ball of the same volume, is it true that 
\[
C_p(K) \geq C_p(B) \ ?
\]
\end{problem}
For background on this problem, see the papers above and the mention by Mattila \cite[p.{\,}193]{M90}.

\section{\bf Potential theoretic background}
\label{sec-background}

This section provides references for the existence and uniqueness of equilibrium measures, and briefly discusses the physical interpretation of equilibrium measure and capacity. 

\subsection*{Existence of the equilibrium measure} A straightforward compactness argument applied to the collection of unit Borel measures on $K$ yields the existence of a measure achieving the minimum in the definition of the energy, for both the logarithmic and Riesz situations \cite[pp.{\,}131--132]{L72}, and indeed for much more general kernels too \cite[Lemma 4.1.3]{BHS19}. 

\subsection*{Uniqueness of the equilibrium measure} Uniqueness of the logarithmic equilibrium measure is well known in $2$ dimensions, on sets of positive logarithmic capacity. A proof can be found, for example, in the monograph by Landkof \cite[pp.\,133,167--168]{L72}, with an appealing variation in Saff and Totik \cite[Theorem I.1.3, Lemma I.1.8]{ST97}. 

In higher dimensions, Anderson and Vamanamurthy \cite[p.\,3 and Lemma 1]{AV88} observe that the proof in Landkof can be adapted to all dimensions. Cegrell, Kolodziej and Levenberg \cite[Theorem 2.5]{CKL98} proved the key result \cite[Lemma 1]{AV88} in detail, namely, that if a signed, compactly supported measure has net mass zero, then its logarithmic energy is nonnegative and the energy equals $0$ if and only if the measure vanishes everywhere. (Those authors go further, and handle measures of unbounded support subject to a growth bound.) Given this result, the uniqueness of logarithmic equilibrium measure follows quickly as in the planar case, by considering two equilibrium measures  and taking their average to deduce that the measures must in fact agree. Incidentally, many geometric properties of logarithmic capacity in higher dimensions have been investigated recently by Xiao \cite{X18,X20}.

In all dimensions, a modern treatment of uniqueness for logarithmic equilibrium measure is given by Borodachov, Hardin and Saff \cite[Theorem 4.4.8]{BHS19}. They also treat uniqueness of the equilibrium measure for Riesz $p$-capacity, when $0<p<n$, on compact sets of positive capacity \cite[Theorem 4.4.5]{BHS19}, and their approach can handle even more general families of kernels. The Riesz case is of course a standard result \cite[pp.{\,}132--133]{L72}. 

\subsection*{Formulas for capacities of special sets} Rather few explicit formulas are known for capacity. Results for various special sets can be found in \cite[pp.{\,}429, 434--436]{H89} and \cite[pp.{\,}163, 165--167, 172--173]{L72}. Notably, no formula is known for the capacity of a cube in $3$ dimensions. 

\subsection*{Physical interpretation of energy} The Newtonian energy $V_1(K)$ in dimension $n=3$ represents the least electrostatic energy (work) required to bring in one total unit of positive charges from infinity and place them on a conductor of shape $K$. The charges are then kept in place by their mutual repulsion. The equilibrium measure describes this lowest-energy distribution of charges. 

Logarithmic energy in $2$ dimensions can be interpreted similarly, by extending the charge distribution uniformly in the vertical direction and measuring the energy per unit length; see \cite[Section 1.5]{L93thesis}. 

\subsection*{Variational capacities} The generalized Newtonian capacity ($p=n-2$) can be characterized in terms of minimizing the Dirichlet integral $\int |\nabla u|^2 \, dx$ of a function $u$ that equals $1$ on $K$ and vanishes at infinity. In other words, the Riesz $(n-2)$-capacity equals the variational $2$-capacity. A direct relationship between $p$-Riesz and $\alpha$-variational capacity ($\min_u \int |\nabla u|^\alpha \, dx$) does not seem to exist for other values of $p$. For more about variational capacities, see Adams and Hedberg \cite{AH96}.

\section{\bf Differentiating the energy functional, for a fixed measure}
\label{sec:differentiating}

We work here with a general family of kernels, so that the logarithmic and $p$-energies can be handled in a unified manner. Throughout the section, $\Phi(r)$ is a smooth, real-valued function for $r>0$, with $\Phi(r) \to \infty$ as $r \to 0$. Let $\Phi(0)=\infty$. In the lemmas that follow, we employ certain assumptions on the kernel:
\begin{align} 
|\Phi(ar)| & \leq C(1+|\Phi(r)|) , \label{Phiassumption1} \\
|\Phi^\prime(ar)| & \leq Cr^{-1} (1+|\Phi(r)|) , \label{Phiassumption2} \\
|\Phi^{\prime\prime}(ar)| & \leq Cr^{-2} (1+|\Phi(r)|) , \label{Phiassumption3}
\end{align}
for all $r > 0$ and $1/2 \leq a \leq 2$, for some constant $C>0$. 

These conditions are easily verified for the logarithmic kernel $\Phi(r)=\log 1/r$ with constant $C=4$, and for the Riesz kernel $\Phi(r)=1/r^p$ with $C=2^{p+2} (1+p+p^2)$. 

Suppose $T$ is an open interval such that 
\[
F_t : \Rn \to \Rn
\]
is a smooth diffeomorphism for each $t \in T$ and $F_t(x)$ is jointly smooth as a function of $(t,x) \in T \times \Rn$. Thus $F_t$ describes a flow on $\Rn$. 

Fix a compact set $K \subset \Rn, n \geq 2$, and a Borel measure $\mu$ on $K$ with $0 < \mu(K) < \infty$. Define the $\Phi$-energy of $\mu$ to be the function 
\[
E_\mu(t) = \int_K \int_K \Phi(|F_t(x)-F_t(y)|) \, d\mu(x) d\mu(y) , \qquad t \in T,
\]
so that $-\infty < E_\mu(t) \leq +\infty$. The next four lemmas prove finiteness, continuity, differentiability, and twice differentiability of $E_\mu$ with respect to $t$. The fifth lemma examines concavity. 
\begin{lemma}[Finiteness of $E_\mu$ somewhere implies finiteness everywhere] \label{le:claim0}
If the kernel $\Phi$ satisfies \eqref{Phiassumption1} and $E_\mu(t)<\infty$ for some $t \in T$, then $E_\mu(t)<\infty$ for all $t \in T$.
\end{lemma}
\begin{proof}
Fix $\tau \in T$. The distortions of the diffeomorphisms $F_{t_1}$ and $F_{t_2}$ are comparable when $t_1$ and $t_2$ are sufficiently close to $\tau$, meaning that
\[
\frac{1}{2} \leq \frac{|F_{t_1}(x)-F_{t_1}(y)|}{|F_{t_2}(x)-F_{t_2}(y)|} \leq 2 , \qquad x,y \in K , \ x \neq y .
\]
Hence hypothesis \eqref{Phiassumption1} implies  
\[
|\Phi(|F_{t_1}(x)-F_{t_1}(y)|)| \leq C + C \, |\Phi(|F_{t_2}(x)-F_{t_2}(y)|)|  .
\]
The absolute values can be dropped on the right side provided the first constant $C$ is increased suitably, since $\Phi(r)$ is bounded below when $r=|F_{t_2}(x)-F_{t_2}(y)|$ is bounded, which it certainly is for $x$ and $y$ belonging to the compact set $K$. Hence 
\begin{equation} \label{eq:Phiestimate}
|\Phi(|F_{t_1}(x)-F_{t_1}(y)|)| \leq A + B \, \Phi(|F_{t_2}(x)-F_{t_2}(y)|) 
\end{equation}
whenever $x,y \in K$ and $t_1$ and $t_2$ are close to $\tau \in T$. The constants $A$ and $B$ depend on $K$ and $\tau$. 

By integrating the last inequality with respect to $d\mu(x)d\mu(y)$ we see that if $E_\mu(t_2)$ is finite then so is $E_\mu(t_1)$, and that if $E_\mu(t_1)$ is infinite then so is $E_\mu(t_2)$. In particular, if $E_\mu(\tau)<\infty$ then $E_\mu(t)<\infty$ for all $t$ near $\tau$, while if $E_\mu(\tau)=\infty$ then $E_\mu(t)=\infty$ for all $t$ near $\tau$. 

Since $\tau$ was arbitrary, we deduce the set $\{ t \in T : E_\mu(t) < \infty \}$ is both open and closed, and so it is either empty or else equals the full interval $T$. The lemma follows. 
\end{proof}
\begin{lemma}[Continuity of $E_\mu$] \label{le:claim1}
If the kernel $\Phi$ satisfies \eqref{Phiassumption1} and $E_\mu(t)<\infty$ for some $t \in T$, then $E_\mu(t)$ is finite valued and continuous on the interval $T$.
\end{lemma}
\begin{proof}
The finiteness of $E_\mu$ was shown in \autoref{le:claim0}. Take $t_2=\tau \in T$. The right side of \eqref{eq:Phiestimate} is integrable with respect to $d\mu(x)d\mu(y)$, since $E_\mu(t_2)<\infty$, and so it provides an integrable dominator for the integral defining $E_\mu(t_1)$. Dominated convergence therefore implies $E_\mu(t_1) \to E_\mu(t_2)$ as $t_1 \to t_2$, which proves the desired continuity of $E_\mu$. 
\end{proof}
\begin{lemma}[Differentiability of $E_\mu$] \label{le:claim2}
If the kernel $\Phi$ satisfies \eqref{Phiassumption1} and \eqref{Phiassumption2}, and $E_\mu(t)<\infty$ for some $t \in T$, then $E_\mu(t)$ is finite valued and differentiable on $T$, and differentiation through the integral holds:
\[
E_\mu^\prime(t) = \int_K \int_K \left( \frac{\partial\ }{\partial t} \, \Phi(|F_t(x)-F_t(y)|) \right) d\mu(x) d\mu(y) .
\]
\end{lemma}
\begin{proof}
Finiteness of the integral defining $E_\mu(t)$ (shown in \autoref{le:claim0}) implies that the set where the integrand $\Phi(|F_t(x)-F_t(y)|)$ equals $\infty$ has $(\mu \times \mu)$-measure equal to zero. That set is exactly the spatial diagonal $\{ (x,y) : x=y \}$, and so we may assume in what follows that $x \neq y$. 

To prove the lemma, one wants to differentiate through the integral using a standard argument with difference quotients and dominated convergence. To justify this step it suffices to take $\tau \in T$ and demonstrate an integrable dominator for the $t$-derivative of the integrand such that the dominator is independent of $t$ near $\tau$. 

Start by computing the derivative directly as 
\begin{align}
& \frac{\partial\ }{\partial t} \, \Phi(|F_t(x)-F_t(y)|) \notag \\
& = \Phi^\prime(|F_t(x)-F_t(y)|) \frac{F_t(x)-F_t(y)}{|F_t(x)-F_t(y)|} \cdot (\dot{F}_t(x)-\dot{F}_t(y)) , \label{firstderivPhi}
\end{align}
where the dot indicates a $t$-derivative. Hence when $t$ is close to $\tau$ we have 
\begin{align*}
& \left| \frac{\partial\ }{\partial t} \, \Phi(|F_t(x)-F_t(y)|) \right| \notag \\
& \leq C \big( 1+|\Phi(|F_t(x)-F_t(y)|) | \big) \, |F_t(x)-F_t(y)|^{-1} \left| \dot{F}_t(x)-\dot{F}_t(y) \right| \quad \ \text{by \eqref{Phiassumption2} with $a=1$} \notag \\
& \leq C \big( 1+|\Phi(|F_t(x)-F_t(y)|)| \big) \qquad \text{since $\dot{F}_t(x)$ is smooth as a function of $(t,x)$} \notag \\
& \leq A + B \, \Phi(|F_{\tau}(x)-F_{\tau}(y)|) 
\end{align*}
by \eqref{eq:Phiestimate}. The last line is independent of $t$, and provides a dominator with respect to $d\mu(x) d\mu(y)$ since $E_\mu(\tau)<\infty$. Thus the lemma is proved. 
\end{proof}
\begin{lemma}[Twice differentiability of $E_\mu$] \label{le:claim3}
If the kernel $\Phi$ satisfies \eqref{Phiassumption1}, \eqref{Phiassumption2}, \eqref{Phiassumption3}, and $E_\mu(t)<\infty$ for some $t \in T$, then $E_\mu(t)$ is twice differentiable on $T$ and its second derivative is found by differentiation through the integral:
\[
E_\mu^{\prime \prime}(t) = \int_K \int_K \left( \frac{\partial^2\ }{\partial t^2} \, \Phi(|F_t(x)-F_t(y)|) \right) d\mu(x) d\mu(y) .
\]
\end{lemma}
\begin{proof}
Arguing as in the proof of \autoref{le:claim2}, the task is to find an integrable dominator for the second $t$-derivative of the integrand such that the dominator is independent of $t$ near $\tau$. Direct computation reveals the second derivative to be 
\begin{align}
& \frac{\partial^2\ }{\partial t^2} \, \Phi(|F_t(x)-F_t(y)|) \notag \\
& = \left( \Phi^{\prime\prime}(r)-\frac{\Phi^\prime(r)}{r} \right) \left( \frac{F_t(x)-F_t(y)}{|F_t(x)-F_t(y)|} \cdot (\dot{F}_t(x)-\dot{F}_t(y)) \right)^{\! 2} \label{eq:gformula} \\
& \qquad + \frac{\Phi^\prime(r)}{r} \left( |\dot{F}_t(x)-\dot{F}_t(y)|^2 + (F_t(x)-F_t(y)) \cdot (\ddot{F}_t(x)-\ddot{F}_t(y)) \right) , \notag 
\end{align}
where for brevity we have written $r=|F_t(x)-F_t(y)|$. Hence when $t$ is close to $\tau$, 
\[
\begin{split}
& \left| \frac{\partial^2\ }{\partial t^2} \, \Phi(|F_t(x)-F_t(y)|) \right| \\
& \leq C (1+|\Phi(r)|) \, r^{-2} \left( |\dot{F}_{t}(x)-\dot{F}_{t}(y)|^2 + \lvert F_{t}(x)-F_{t}(y) \rvert \lvert \ddot{F}_{t}(x)-\ddot{F}_{t}(y) \rvert \right) 
\end{split}
\]
by using hypotheses \eqref{Phiassumption2} and \eqref{Phiassumption3} with $a=1$. Since $F_t(x), \dot{F}_t(x), \ddot{F}_t(x)$ are smooth as functions of $(t,x)$, we deduce for $x,y \in K$ that
\[
\left| \frac{\partial^2\ }{\partial t^2} \, \Phi(|F_t(x)-F_t(y)|) \right| \leq C (1+|\Phi(r)|) \leq A + B \Phi(\rho) 
\]
where $\rho=|F_\tau(x)-F_\tau(y)|$ and the second inequality relies on the distortion estimate \eqref{eq:Phiestimate}, which in turn relies on hypothesis \eqref{Phiassumption1}.

The last estimate is independent of $t$, and provides a dominator with respect to $d\mu(x) d\mu(y)$ since $E_\mu(\tau)<\infty$, completing the proof of the lemma. 
\end{proof}
\begin{lemma}[Concavity of $E_\mu$] \label{le:claim4}
Suppose the kernel $\Phi$ satisfies conditions \eqref{Phiassumption1}, \eqref{Phiassumption2}, \eqref{Phiassumption3}, and that $E_\mu(t)<\infty$ for some $t \in T$. 

If $\Phi^{\prime \prime}(r) \geq r^{-1} \Phi^\prime(r)$ for $r>0$, and 
\begin{equation} \label{bigcondition}
\begin{split}
0 & \geq \Phi^{\prime \prime}(|F_t(x)-F_t(y)|) \, |\dot{F}_t(x)-\dot{F}_t(y)|^2 \\
& \quad + \frac{\Phi^\prime(|F_t(x)-F_t(y)|)}{|F_t(x)-F_t(y)|} (F_t(x)-F_t(y)) \cdot (\ddot{F}_t(x)-\ddot{F}_t(y)) 
\end{split}
\end{equation}
for all $x,y \in K, x \neq y$ and all $t \in T$, then $E_\mu(t)$ is concave, with second derivative $E_\mu^{\prime\prime}(t) \leq 0$ for each $t \in T$. The second derivative is negative at $t$ if in addition: 
\begin{enumerate}
\item[(i)] inequality \eqref{bigcondition} is strict, or else 
\item[(ii)] $\Phi^{\prime\prime}(r) > r^{-1} \Phi^\prime(r)$ for all $r>0$, and the set 
\[
\{ (x,y) \in K \times K : \text{$F_t(x)-F_t(y)$ and $\dot{F}_t(x)-\dot{F}_t(y)$ are linearly independent} \}
\]
has positive $(\mu \times \mu)$-measure. 
\end{enumerate}
\end{lemma}
\begin{proof}
The Schwarz inequality gives that 
\begin{equation} \label{eq:schwarz}
\left( \frac{F_t(x)-F_t(y)}{|F_t(x)-F_t(y)|} \cdot (\dot{F}_t(x)-\dot{F}_t(y)) \right)^{\! 2} \leq |\dot{F}_t(x)-\dot{F}_t(y)|^2 
\end{equation}
when $x \neq y$, with strict inequality when the vectors $F_t(x)-F_t(y)$ and $\dot{F}_t(x)-\dot{F}_t(y)$ are linearly independent. After inserting this inequality into the formula \eqref{eq:gformula}, and using the assumption that $\Phi^{\prime\prime}(r)-\Phi^\prime(r)/r \geq 0$, we can simplify to obtain that  
\begin{align}
& \frac{\partial^2\ }{\partial t^2} \, \Phi(|F_t(x)-F_t(y)|) \notag \\
& \leq \Phi^{\prime\prime}(r) |\dot{F}_t(x)-\dot{F}_t(y)|^2 + \frac{\Phi^\prime(r)}{r} (F_t(x)-F_t(y)) \cdot (\ddot{F}_t(x)-\ddot{F}_t(y)) \label{eq:schwarzstrict} \\
& \leq 0 \notag 
\end{align}
by hypothesis \eqref{bigcondition}. Thus the map $t \mapsto \Phi(|F_t(x)-F_t(y)|)$ is concave, for each pair of points $x \neq y$. We need not consider $x=y$, since the spatial diagonal set has $(\mu \times \mu)$-measure zero, due to finiteness of $E_\mu(t)$. Concavity of $E_\mu(t)$ now follows from \autoref{le:claim3}.

Strictness of the concavity is immediate by the argument above when condition (i) holds, that is, when hypothesis \eqref{bigcondition} holds with strict inequality for all $x \neq y$. 

Now suppose condition (ii) holds, so that $\Phi^{\prime\prime}(r) - r^{-1} \Phi^\prime(r)>0$, and the Schwarz inequality \eqref{eq:schwarz} holds with strict inequality for $(x,y)$ in some set of positive $(\mu \times \mu)$-measure. On that set, inequality \eqref{eq:schwarzstrict} holds with strict inequality, and so \autoref{le:claim3} yields negativity of $E_\mu^{\prime\prime}(t)$. 
\end{proof}

\section{\bf Linear maps, for a fixed measure}
\label{sec:logvariation}

Now we specialize to diagonal linear diffeomorphisms and prove that the logarithmic and Riesz energies are concave with respect to the variation parameter $t$. Further, the energy has a critical point at $t=0$ if the measure $\mu$ is invariant under an irreducible group of isometries. 

Throughout the section, $K$ is a compact set in $\Rn, n \geq 2$. Fix numbers 
\[
\sigma_1,\dots,\sigma_n>0 .
\] 

For the logarithmic case ($p=0$), assume  
\begin{equation} \label{singprod}
\sigma_1 \cdots \sigma_n = 1 ,
\end{equation}
and define a $1$-parameter family of diagonal matrices by
\begin{equation} \label{Ttlogarithmic}
S_t = 
\begin{pmatrix}
\sigma_1^t & 0 & \dots & 0 \\ 0 & \sigma_2^t & \dots & 0 \\ \vdots & \vdots & \ddots & \vdots \\ 0 & 0 & \dots & \sigma_n^t
\end{pmatrix}
\end{equation}
for $t \in \R$. Notice $S_0$ is the identity and $S_1$ is a diagonal matrix with the $\sigma_k$ on the diagonal. Each $S_t$ is a volume-preserving stretch of $n$-dimensional space, since $\det S_t = 1$. Its derivative at $t=0$ is 
\[
\dot{S}_0 = 
\begin{pmatrix}
\log \sigma_1 & 0 & \dots & 0 \\ 0 & \log \sigma_2 & \dots & 0 \\ \vdots & \vdots & \ddots & \vdots \\ 0 & 0 & \dots & \log \sigma_n
\end{pmatrix} .
\]
Clearly $\dot{S}_0$ has trace zero since $\sigma_1 \cdots \sigma_n = 1$.  The choice of $S_t$ will be motivated at the end of the section.  

For the Riesz case ($0<p<n$), assume 
\begin{equation} \label{sigmacond}
\frac{1}{n} ( \sigma_1^{-p} + \dots + \sigma_n^{-p} ) = 1 ,
\end{equation}
and define 
\[
S_t =
\begin{pmatrix}
(1-t+t\sigma_1^{-p})^{-1/p} & 0 & \dots & 0 \\ 0 & (1-t+t\sigma_2^{-p})^{-1/p} & \dots & 0 \\ \vdots & \vdots & \ddots & \vdots \\ 0 & 0 & \dots & (1-t+t\sigma_n^{-p})^{-1/p}
\end{pmatrix}
\]
for $t \in [0,1]$. The definition continues to be valid for $t$ in a slightly larger open interval $T$ that contains $[0,1]$. Again $S_0$ is the identity and $S_1$ is a diagonal matrix with the $\sigma_k$ on the diagonal. The derivative at $t=0$ is 
\[
\dot{S}_0 = 
\frac{1}{p} 
\begin{pmatrix}
1-\sigma_1^{-p} & 0 & \dots & 0 \\ 0 & 1-\sigma_2^{-p}  & \dots & 0 \\ \vdots & \vdots & \ddots & \vdots \\ 0 & 0 & \dots & 1-\sigma_n^{-p} 
\end{pmatrix} .
\]
Observe $\dot{S}_0$ has trace zero by \eqref{sigmacond}. 

Consider the linear diffeomorphism $F_t(x)=S_t x$ on $\Rn$, and work from now on with the logarithmic and Riesz kernels  
\[
\Phi(r) = 
\begin{cases}
\log 1/r & \text{when $p=0$,} \\
1/r^p & \text{when $0<p<n$,}
\end{cases}
\]
which are known to satisfy assumptions \eqref{Phiassumption1}, \eqref{Phiassumption2}, \eqref{Phiassumption3}. We start by showing that the first derivative of the energy 
\[
E_\mu(t) = \int_K \int_K \Phi(|S_t(x)-S_t(y)|) \, d\mu(x) d\mu(y) 
\]
vanishes at $t=0$, when $\mu$ possesses sufficient symmetry. 
\begin{proposition}[First variation $=0$ for symmetric measures] \label{logcapfirstderivsymmetric} 
Assume $\mu$ is a Borel measure on $K$ with $0 < \mu(K) < \infty$, and $E_\mu(0)<\infty$. If $\mu$ is invariant under an irreducible, compact group of isometries then $E_\mu^\prime(0) = 0$.
\end{proposition}
\begin{proof}
The first derivative at $t=0$ can be evaluated by applying \autoref{le:claim2}, that is, by substituting the definition of the logarithmic or Riesz kernel $\Phi$ into formula \eqref{firstderivPhi}, obtaining that 
\begin{equation} \label{firstderivatzero}
E_\mu^\prime(0)
= - \beta_p \int \! \int \frac{(x-y) \cdot \dot{S}_0(x-y)}{|x-y|^{p+2}} \, d\mu(x) d\mu(y) 
\end{equation}
where we used that $S_0$ is the identity map and $\dot{S}_0$ is linear; the constant factor on the right side is 
\[
\beta_p = 
\begin{cases}
1 & \text{when $p=0$,} \\
p & \text{when $0<p<n$.}
\end{cases}
\] 
Write ${\mathcal U}$ for the irreducible, compact group of isometries under which $\mu$ is invariant, so that $\mu = \mu \circ U^{-1}$ for each $U \in {\mathcal U}$. Replacing $\mu$ with $\mu \circ U^{-1}$ in the right side of the derivative formula \eqref{firstderivatzero} yields that 
\[
E_\mu^\prime(0)
= - \beta_p \int \! \int \frac{(x-y)^\dagger U^\dagger \dot{S}_0 U(x-y)}{|x-y|^{p+2}} \, d\mu(x) d\mu(y) ,
\]
where $\dagger$ denotes the matrix transpose. The left side is independent of $U$. Integrating the right side over $U \in \mathcal{U}$ with respect to Haar measure gives $0$, by \autoref{averagingnD}, since $\dot{S}_0$ is real and symmetric with trace zero as we observed earlier in the section. 
\end{proof}

Next we show the logarithmic and Riesz energies are concave with respect to the variation parameter. 
\begin{proposition}[Second variation of the energy is $\leq 0$] \label{concavitypenergy} 
Assume $\mu$ is a Borel measure on $K$ with $0 < \mu(K) < \infty$. If $E_\mu(0)<\infty$ then $E_\mu^{\prime \prime}(t) \leq 0$ for all $t \in \R$ (when $p=0$) or all $0 \leq t \leq 1$ (when $0<p<n$). The strict inequality $E_\mu^{\prime \prime}(0) < 0$ holds at $t=0$ if in addition $\sigma_j \neq 1$ for some $j$ and $\mu$ has less than full measure on each hyperplane $\{ x_k = \text{const} \}$.
\end{proposition}
\begin{proof}
Concavity will be obtained from \autoref{le:claim4}. Obviously $\Phi^{\prime \prime}(r) > r^{-1} \Phi^\prime(r)$, since the logarithmic and Riesz kernels are strictly convex and decreasing. Substituting the definition of $\Phi$ into the desired inequality \eqref{bigcondition} for \autoref{le:claim4}, we find it is equivalent to
\[
(p+1) |\dot{S}_t(x)-\dot{S}_t(y)|^2 \leq (S_t(x)-S_t(y)) \cdot (\ddot{S}_t(x)-\ddot{S}_t(y)) , \qquad x \neq y .
\]
Since $S_t$ is linear, the condition simplifies to
\begin{equation} \label{Sinequality}
(p+1) |\dot{S}_t x|^2 \leq S_t x \cdot \ddot{S}_t x , \qquad x \neq 0 .
\end{equation}
Recalling the definition \eqref{Ttlogarithmic} of the diagonal matrix $S_t$ in the logarithmic case ($p=0$), we find that in fact equality holds, because both sides equal $\sum_{k=1}^n (\sigma_k^t \log \sigma_k)^2 x_k^2$. Similarly, in the Riesz case both sides equal 
\[
\frac{p+1}{p^2} \sum_{k=1}^n (1-t+t\sigma_k^{-p})^{-2-2/p} (\sigma_k^{-p}-1)^2 x_k^2 .
\]
Hence \autoref{le:claim4} yields that $E_\mu^{\prime \prime}(t) \leq 0$. 

To obtain strict inequality at $t=0$, under the assumptions that $\sigma_j \neq 1$ for some $j$ and $\mu$ has less than full measure on each hyperplane $\{ x_k = \text{const} \}$, we will verify condition (ii) in \autoref{le:claim4} with $t=0$, which means we want to show that the set 
\[
L = \{ (x,y) : \text{$x-y$ and $\dot{S}_0(x-y)$ are linearly independent} \}
\]
has positive $(\mu \times \mu)$-measure. We will prove a contrapositive statement: if $\sigma_j \neq 1$ for some $j$ and 
\[
(\mu \times \mu) (L) = 0
\]
then $\mu$ has full measure on some hyperplane $\{ x_k = \text{const} \}$. 

First consider the logarithmic case, $p=0$. Observe that $\sigma_j \neq \sigma_k$ for some $j \neq k$, since $\sigma_1 \cdots \sigma_n = 1$ and $\sigma_j \neq 1$ for some $j$. For these choices of $j,k$, we claim that $(x,y) \in L$ whenever $x_j \neq y_j$ and $x_k \neq y_k$. Indeed, considering only the $j$-th and $k$-th components of the vectors $x-y$ and $\dot{S}_0(x-y)$ gives vectors
\[
\begin{pmatrix} x_j-y_j \\ x_k-y_k \end{pmatrix}
 \qquad \text{and} \qquad 
\begin{pmatrix} (x_j-y_j) \log \sigma_j \\ (x_k-y_k) \log \sigma_k \end{pmatrix} ,
\]
which are linearly independent since $x_j-y_j \neq 0, x_k-y_k \neq 0$ and $\sigma_j \neq \sigma_k$.

The last paragraph implies that  
\[
(\mu \times \mu) \big( \{ (x,y) : x_j \neq y_j \text{\ and\ } x_k \neq y_k \} \big) \leq (\mu \times \mu) (L) = 0 .
\]
Taking cross-sections, we deduce $\mu(\{ x : x_j \neq y_j \text{\ and\ } x_k \neq y_k \}) = 0$ for $\mu$-almost every $y \in \Rn$. Fix a point $y$ for which this equation holds, and partition $\Rn$ into the disjoint sets 
\begin{align*}
A & = \{ x : x_j \neq y_j \text{\ and\ } x_k \neq y_k \} , \\
B & = \{ x : x_j = y_j \text{\ and\ } x_k = y_k \} , \\
C & = \{ x : x_j = y_j \text{\ and\ } x_k \neq y_k \} , \\
D & = \{ x : x_j \neq y_j \text{\ and\ } x_k = y_k \} .
\end{align*}
Then $\mu(A)=0$ by our choice of the point $y$, and so $\mu(B \cup C \cup D)=\mu(K)$.

If $x \in C$ and $z \in D$, then $x_j = y_j \neq z_j$ and $x_k \neq y_k = z_k$. Hence $(x,z) \in L$, by above, which means $C \times D \subset L$. Therefore $\mu(C)\mu(D)=0$, which means either $C$ or $D$ has $\mu$-measure zero. Suppose for the sake of definiteness that $\mu(C)=0$. Then $\mu(B \cup D)=\mu(K)$. That is, $\mu$ has full measure on the hyperplane $B \cup D = \{  x : x_k= y_k \}$, which completes the proof of the contrapositive. 

Next consider the Riesz case, $0<p<n$. The argument to complete the strict inequality statement proceeds as for the logarithmic case above: $\sigma_j \neq \sigma_k$ for some $j \neq k$, since $( \sigma_1^{-p} + \dots + \sigma_n^{-p} )/n = 1$ and $\sigma_j \neq 1$ for some $j$. From this it follows that $(x,y) \in L$ whenever $x_j \neq y_j$ and $x_k \neq y_k$. Considering only the $j$-th and $k$-th components of the vectors $x-y$ and $\dot{S}_0(x-y)$ gives 
\[
\begin{pmatrix} x_j-y_j \\ x_k-y_k  \end{pmatrix}
 \quad \text{and} \quad 
\frac{1}{p}\begin{pmatrix} (x_j-y_j) (1 - \sigma_j^{-p}) \\ (x_k-y_k) (1 - \sigma_k^{-p}) \end{pmatrix} ,
\]
which are easily checked to be linearly independent since $\sigma_j \neq \sigma_k$.
\end{proof}
\noindent \textbf{Motivation for $S_t$.} The motivation for choosing $S_t$ earlier in the section is revealed by the proof above. Each diagonal entry $d(t)$ of $S_t$ is chosen to satisfy $(p+1)(d^\prime)^2 = d d^{\prime\prime}$ in order to satisfy \eqref{Sinequality} with equality. The equation has exponential solutions when $p=0$, explaining our choice of the diagonal entry $\sigma^t$ in the matrix $S_t$ for the logarithmic case, and in the Riesz case $0<p<n$ the equation has solutions of the form $(a+bt)^{-1/p}$.

\medskip
\noindent \textbf{Connection to Schiffer's work.} Concavity of the energy $E_\mu(t)$ was observed by Schiffer \cite[p.{\,}321]{S54} in the Newtonian case $p=1$ for the  family of reciprocal linear stretches
\[
F_t =
\begin{pmatrix}
1/t & 0 & 0 \\ 0 & 1 & 0 \\ 0 & 0 & 1
\end{pmatrix} .
\]
His result fits the form given above with $p=1$ and $d(t)=1/t$. By the way, Schiffer actually showed concavity of $t E_\mu(1/t)$, which is equivalent to concavity of $E_\mu(t)$.

\section{\bf Proof of \autoref{capacitymin} and \autoref{capacityminp}}
\label{mainproof}

The idea is to reduce to a diagonal matrix by the singular value decompositon, and then apply linear deformations and energy concavity results from the preceding section. 

The singular value decomposition says $M=A S B$, where $A$ and $B$ are orthogonal matrices and $S$ is diagonal with the singular values $\sigma_1,\dots,\sigma_n$ on the diagonal. We may assume $A$ is the identity, because capacity is invariant under orthogonal transformations. We may further reduce to the case $B=I$, by considering the rotated set $\widetilde{K}=BK$, which is again compact with an irreducible isometry group. Thus we may suppose $M$ equals the diagonal matrix $S$.

If $S$ is the identity matrix then there is nothing to prove. So assume $S$ is not the identity. The task is to prove that $\logcap(SK) > \logcap(K)$ in the logarithmic case and $C_p(SK) > C_p(K)$ in the Riesz case.

For logarithmic capacity in \autoref{capacitymin}, the matrix $M$ is assumed to be volume-preserving, which means the product of its singular values equals $1$. Thus assumption \eqref{singprod} holds. For Riesz capacity, the hypothesis in \autoref{capacityminp} on the $p$-norm of the singular values of $M^{-1}$ says exactly that the assumption \eqref{sigmacond} holds. Hence in each case, the results of \autoref{sec:logvariation} can be applied with the linear diffeomorphism $S_t$. Let 
\[
K_t = S_t(K) , \qquad 0 \leq t \leq 1 ,
\]
be the image of $K$ under the stretch map $S_t$, and write $V$ for either $\logV$ or $V_p$, because the proof that follows is the same in both cases. 

Since $S_0$ is the identity and $S_1=S$, we have $K_0=K$ and $K_1=SK$. Therefore it will suffice to show 
\begin{equation} \label{Vinequalities}
V(K_0) > V(K_t) \qquad \text{when $0 < t \leq 1$}
\end{equation}
and $K$ is a compact set having finite energy $V(K)<\infty$ and irreducible isometry group.  Taking $t=1$ yields the theorems. 

The energy of $K_t$ is
\begin{align*}
V(K_t) 
& = \min_\nu \int_{K_t} \int_{K_t} \Phi(|x-y|) \, d\nu(x) d\nu(y) \\
& \leq \int_K \int_K \Phi(|S_t x - S_t y|) \, d\mu(x) d\mu(y) = E_\mu(t) ,
\end{align*}
by choosing $\nu=\mu \, \circ \, S_t^{-1}$ to be the pushforward of $\mu$ under $S_t$, where $\mu$ is the equilibrium measure on $K$. Equality holds at $t=0$, and so \eqref{Vinequalities} will follow once we show
\[
E_\mu(0) > E_\mu(t) \qquad \text{whenever $0 < t \leq 1$.}
\]
That is, we want $E_\mu(t)$ to be strictly maximal at $t=0$. 

To prove this assertion we plan to invoke \autoref{logcapfirstderivsymmetric} and \autoref{concavitypenergy}, and so we must verify their hypotheses. We know $E_\mu(0)=V(K)<\infty$ because $K$ has positive capacity by assumption. The isometry group of $K$ is compact because $K$ is compact. Also, $\sigma_j \neq 1$ for some $j$ because $S$ is not the identity matrix. 

Suppose (to the contrary of what we wish to establish) that the equilibrium measure $\mu$ of $K$ has full measure on some hyperplane $H = \{ x \in \Rn : x^\dagger y = s \}$, where $y$ is the normal vector and $s \in \R$. That is, suppose $\mu(\Rn \setminus H) = 0$. By the irreducibility hypothesis, isometries $U_1,\dots,U_n$ of $K$ exist for which $\{U_1 y,\dots,U_n y\}$ forms a basis in $\Rn$. Thus the square matrix $W = [U_1y \cdots U_n y]$ is invertible. Write $E = U_1(H) \cap \dots \cap U_n(H)$ for the intersection of the images of the hyperplane under the chosen isometries. This set consists of a single point, since $x \in E$ means $x^\dagger U_k y = s$ for all $k$, which is equivalent to $x^\dagger W = [s \dots s]$, and this last equation determines $x$ uniquely by invertibility of $W$; hence $E = \{ x \}$ consists of exactly that one point. Further,  
\begin{align*}
\mu(\Rn \setminus E)
& = \mu \big( \cup_{k=1}^n \Rn \setminus U_k(H) \big) \\
& \leq \sum_{k=1}^n \mu \big( U_k(\Rn \setminus H) \big) \\
& = 0 
\end{align*}
because $\mu(\Rn \setminus H) = 0$ and uniqueness of the equilibrium measure implies invariance of $\mu$ under each isometry $U_k$. The inequality shows $\mu$ has full measure on the single-point set $E=\{ x \}$, which means $\mu$ is a multiple of a delta measure at that point and hence $V(K)=\infty$, contradicting our hypothesis. We conclude that $\mu$ has less than full measure in every $(n-1)$-dimensional hyperplane. 

\autoref{concavitypenergy} now says that the energy $E_\mu(t)$ is a concave function of $t$, with its second derivative negative at $t=0$, and the first derivative equal to $0$ at $t=0$ by \autoref{logcapfirstderivsymmetric}. Hence $E_\mu$ is strictly maximal at $t=0$, as we needed to show.

\section{\bf Proof of \autoref{capmoment}}
\label{capmomentproof}

Let $M$ be a real, invertible $n \times n$ matrix. The proof involves estimating the $p$-norm of $M^{-1}$ in terms of the $2$-norm of $M$, and then relating that $2$-norm to the moment of inertia and volume. 

Let $q = 2/(n-1)$, so that $0 < p \leq q$ by assumption in the corollary. When $n=2$ it is further assumed that $0 < p < 2 = q$, but that stricter condition will be used only to ensure that the $p$-capacity is well defined. 

Since $p \leq q$, Jensen's inequality yields that 
\begin{equation} \label{Jensen}
\lVert M^{-1} \rVert_{p,n} \leq \lVert M^{-1} \rVert_{q,n}.
\end{equation}

Next we show 
\begin{equation} \label{Schattenest}
\lVert M^{-1} \rVert_{q,n} \leq \lVert M \rVert_{2,n}^{n-1} / |\det M| .
\end{equation}
To prove this claim, write $\sigma_1,\dots,\sigma_n > 0$ for the singular values of $M$, and observe that
\begin{align*}
|\det M| \, \lVert M^{-1} \rVert_{q,n}
& = \left( \prod_{j=1}^n \sigma_j \right) \left( \frac{1}{n} \sum_{k=1}^n \sigma_k^{-q} \right)^{\!\! 1/q} \\
& = \left( \frac{1}{n} \sum_{k=1}^n \prod_{j \neq k} \sigma_j^q \right)^{\!\! 1/q} \\
& = \left( \frac{1}{n} \sum_{k=1}^n \prod_{j=1}^{n-1} \sigma_{j+k}^q \right)^{\!\! 1/q} 
\end{align*}
where for notational convenience we have defined $\sigma_{m+n}=\sigma_m$ for $m=1,\dots,n$. Write $f_j(k)=\sigma_{j+k}^q$ in the preceding sum and apply the discrete H\"{o}lder inequality with exponent $2/q$ (using that $(q/2)+\dots +(q/2)=(n-1)(q/2)=1$) to deduce that 
\begin{align*}
|\det M| \, \lVert M^{-1} \rVert_{q,n}
& \leq \prod_{j=1}^{n-1} \left( \frac{1}{n} \sum_{k=1}^n f_j(k)^{2/q} \right)^{\!\! 1/2} \\
& = \prod_{j=1}^{n-1} \left( \frac{1}{n} \sum_{k=1}^n \sigma_k^2 \right)^{\!\! 1/2} =  \lVert M \rVert_{2,n}^{n-1} ,
\end{align*}
which proves \eqref{Schattenest}. 

To relate these observations to the volume and moment of inertia, notice the scale invariant form of \autoref{capacityminp} says that  
\[
C_p(K) \leq C_p(MK) \, \lVert M^{-1} \rVert_{p,n} ,
\]
with equality if and only if $M$ is a nonzero scalar multiple of an orthogonal matrix, that is, if and only if $M$ is a euclidean similarity. Combining this inequality with \eqref{Jensen} and \eqref{Schattenest}, we find
\begin{equation} \label{capacityandM}
C_p(K) \leq C_p(MK) \, \frac{\lVert M \rVert_{2,n}^{n-1}}{|\det M|} ,
\end{equation}
with equality if and only if $M$ is a euclidean similarity. Obviously 
\[
\frac{1}{|\det M|} = \frac{V(K)}{V(MK)} ,
\]
while \autoref{momentformulas} provides the formula 
\[
\lVert M \rVert_{2,n}^2 = \frac{1}{|\det M|} \frac{I(MK)}{I(K)} .
\]
Substituting these expressions into inequality \eqref{capacityandM}, we obtain that
\[
C_p(K) \, \sqrt{\frac{I(K)^{n-1}}{V(K)^{n+1}}} \leq C_p(MK) \, \sqrt{\frac{I(MK)^{n-1}}{V(MK)^{n+1}}} 
\]
with equality if and only if $M$ is a euclidean similarity. This completes the proof of the corollary.

\section{\bf Proof of \autoref{capacityminn-2}}
\label{corproofsec}

The scale invariant form of \autoref{capacityminp} says in the special case $p=2$ that  
\begin{equation} \label{C2estimate}
C_2(MK) \, \lVert M^{-1} \rVert_{2,n} \geq C_2(K) .
\end{equation}
Meanwhile, \autoref{momentformulas} applied to $M^{-1}$ gives that  
\begin{align*}
\lVert M^{-1} \rVert_{2,n}^2 
& = \frac{1}{|\det M^{-1}|} \, \frac{I(M^{-1}K)}{I(K)} \\
& = \frac{V(K)^{1+4/n}}{I(K)} \frac{I(M^{-1}K)}{V(MK)^{2/n} \, V \big( M^{-1}K \big)^{1+2/n}} .
\end{align*}
Taking the square root and remembering the definition \eqref{alphadef} of the asymmetry $\alpha(\cdot)$, we deduce
\begin{equation} \label{2normidentity}
\lVert M^{-1} \rVert_{2,n} = \frac{\alpha(M^{-1} K)/V(MK)^{1/n}}{\alpha(K)/V(K)^{1/n}} .
\end{equation}
The corollary now follows by substituting into \eqref{C2estimate}. 

\section*{Acknowledgments}
This research was supported by grants from the Simons Foundation (\#429422 to Richard Laugesen) and the University of Illinois Research Board (RB19045). I am grateful to Igor Pritsker for help with the literature on logarithmic capacity in higher dimensions, and to Juan Manfredi for pointing out references on variational $p$-capacities.

\appendix

\section{\bf Averaging over isometries}
Matrix averages provide a key tool in proving \autoref{logcapfirstderivsymmetric}. We begin with an elementary result in $2$ dimensions. 
\begin{lemma}[$2$-dimensional rotational averaging] \label{averaging2D}
Let $\theta_j = 2\pi j/N$. If $N \geq 3$ then 
\[
\sum_{j=1}^N 
\begin{pmatrix}
\ \ \cos \theta_j & \sin \theta_j \\ -\sin \theta_j & \cos \theta_j
\end{pmatrix}
\begin{pmatrix}
-1 & 0 \\ \ 0 & 1
\end{pmatrix}
\begin{pmatrix}
\cos \theta_j & -\sin \theta_j \\ \sin \theta_j & \ \ \cos \theta_j 
\end{pmatrix} 
=
\begin{pmatrix}
0 & 0 \\ 0 & 0
\end{pmatrix} .
\]
\end{lemma}
\begin{proof}
Multiplying the matrices on the left side of the formula yields
\[
\sum_{j=1}^N 
\begin{pmatrix}
-\cos 2\theta_j & \sin 2\theta_j \\ \ \ \sin 2\theta_j & \cos 2\theta_j 
\end{pmatrix} ,
\]
and so it suffices to observe that 
\[
\sum_{j=1}^N (\cos 2\theta_j + i \sin 2\theta_j) = \sum_{j=1}^N e^{(4\pi i/N)j} = e^{4\pi i/N} \frac{e^{(4\pi i/N)N} - 1}{e^{4\pi i/N} - 1} = 0 .
\]
\end{proof}
The denominator of the last fraction vanishes when $N=2$, and so the argument breaks down. Indeed, when $N=2$, \autoref{averaging2D} is easily seen to be false. 

\smallskip
Write $\dagger$ for the matrix transpose operation.  
\begin{lemma}[$n$-dimensional rotational averaging] \label{averagingnD}
If $\mathcal{U}$ is an irreducible, compact subgroup of the orthogonal group $\mathcal{O}(n), n \geq 2$, and $\rho$ is Haar  measure on $\mathcal{U}$, then
\[
\int_{\mathcal{U}} U^\dagger S U \, d\rho(U) = 0 
\]
whenever $S$ is a real symmetric $n \times n$ matrix with trace $0$. 
\end{lemma}
If $\mathcal{U}$ is finite then the Haar measure is simply counting measure and the lemma says $\sum_{U \in \mathcal{U}} U^\dagger S U = 0$, which generalizes the $2$-dimensional result in \autoref{averaging2D}. 

\autoref{averagingnD} is a consequence of Schur's Lemma. We include a direct proof below.
\begin{proof}
Let $M = \int_{\mathcal{U}} U^\dagger S U \, d\rho(U)$. The matrix $M$ is symmetric, since $S$ is symmetric. Furthermore, $M$ commutes with each group element $U \in \mathcal{U}$, with $MU=UM$ by the group property of $\mathcal{U}$ and invariance of Haar measure. Let $\lambda$ be an eigenvalue of $M$ with eigenvector $w$. Then
\[
  M(Uw)=U(Mw)=U(\lambda w)=\lambda(Uw).
\]
Hence the entire orbit $\{ Uw : U \in \mathcal{U} \}$ consists of eigenvectors belonging to $\lambda$. The orbit spans all of $\Rn$, by the irreducibility hypothesis, and so $M$ equals $\lambda$ times the identity matrix. Taking the trace yields
\[
\lambda n = \tr M =\int_{\mathcal{U}} \tr(U^\dagger SU) \, d\rho(U) =\int_{\mathcal{U}} \tr(S) \, d\rho(U) = \tr(S).
\]
Since $\tr S = 0$ by hypothesis, we conclude $\lambda=0$ and hence $M=0$, which proves the lemma. \end{proof}

Irreducibility leads also to a simple (and known) formula for the moment of inertia of a linear image, which we employ when proving \autoref{capmoment} and \autoref{capacityminn-2}. 
\begin{lemma} \label{momentformulas}
If $K \subset \Rn, n \geq 2$, is a compact set with irreducible isometry group and $M$ is an invertible $n \times n$ matrix, then
\[
I(MK) = |\det M| \, \lVert M \rVert_{2,n}^2 \, I(K) .
\]
\end{lemma}
\begin{proof}
Write $\mathcal{U}$ for the irreducible isometry group of $K$. Note $\mathcal{U}$ is compact due to compactness of $K$. The centroid of $K$ lies at the origin as a consequence of the irreducibility, and hence the centroid of $MK$ also lies at the origin, by a linear change of variable.

Define the moment matrix of $K$ to be $Q = \int_K x x^\dagger \, dx$, where $x$ is a column vector. We claim $Q$ is a scalar multiple of the identity. For let $U \in \mathcal{U}$ be an isometry of $K$.  The invariance of $K$ under $U$ implies that $Q=U^\dagger QU$, so that $Q = \frac{1}{n} \tr(Q) \, \text{Id.}$ by arguing as in the proof of \autoref{averagingnD}, that is, by using Schur's Lemma again. The definition of $Q$ reveals that its trace equals the moment of inertia of $K$, using  that the centroid of $K$ lies at the origin, and so 
\begin{equation} \label{momentidentity}
Q = \frac{1}{n} I(K) \, \text{Id.}
\end{equation}

The moment of inertia of $MK$ can now be computed as
\begin{align*}
I(MK)
& = \tr \int_{MK} x x^\dagger \, dx && \text{since the centroid of $MK$ lies at the origin} \\
& = \tr \big( MQ M^\dagger |\det M| \big) && \text{by a change of variable $x \mapsto Mx$} \\
& =\frac{1}{n} I(K) \, \big( \tr M M^\dagger \big) |\det M| && \text{by \eqref{momentidentity}} \\
& = I(K) \, \lVert M \rVert_{2,n}^2 \, |\det M| .
\end{align*}
\end{proof}

\end{document}